\documentclass[a4paper]{./svproc}
\usepackage{url}

\usepackage{caption}
\usepackage{graphicx}
\usepackage{amsmath}
\usepackage{{ifthen}}
\usepackage{amssymb}
\usepackage{subfigure}

\let\subparagraph\paragraph
\usepackage{titlesec}
\titlespacing\section{1pt}{6pt plus 0pt minus 1pt}{6pt plus 0pt minus 2pt}
\titlespacing\subsection{1pt}{6pt plus 0pt minus 1pt}{2pt plus 0pt minus 2pt}
\titlespacing\subsubsection{1pt}{6pt plus 0pt minus 1pt}{2pt plus 0pt minus 2pt}

\DeclareMathOperator*{\argmax}{arg\,max}

\newcommand{\R}{\ensuremath{\mathbb{R}}} 

\newcommand{\angled}[1]{\left\langle#1\right\rangle}
\newcommand{\sqBrackets}[1]{\left[#1\right]}

\newcommand{\task}[0]{\ensuremath{a}}

\newcommand{\allocation}[0]{\ensuremath{\textbf{A}}}  
\newcommand{\allocationSpace}[0]{\ensuremath{\boldsymbol{\mathcal{A}}}}

\newcommand{\taTaskNetwork}[0]{\ensuremath{T}}

\newcommand{\setInitialConfigurations}[0]{\ensuremath{I_c}}

\newcommand{\setInitialTerminalConfigurationsTaskNetwork}[0]{\ensuremath{L_\taTaskNetwork}}

\newcommand{\schedule}[1][none]{\possibleBold{#1}{\ensuremath{\sigma}}}

\newcommand{\scheduleSolution}[0]{\hat{\schedule}}

\newcommand{\worldState}[0]{\ensuremath{W}}

\newcommand{\allocationSolution}[0]{\ensuremath{\hat{\allocation}}}
\newcommand{\mpSolution}[0]{\ensuremath{\hat{\setMotionPlans}}}

\newcommand{\solutionNode}{\hat{N}}

\newcommand{\possibleBold}[2]{%
    \ensuremath{%
        \ifthenelse{\equal{#1}{bold}}%
        {\boldsymbol{#2}}%
        {#2}%
    }%
}

\newcommand{\performanceFunction}[1][none]{\possibleBold{#1}{\xi}}
\newcommand{\performanceFunctionSet}[1][none]{\possibleBold{#1}{\Xi}}
\newcommand{\maxMakespan}[1][none]{\possibleBold{#1}{\ensuremath{C}_{max}}}
\newcommand{\taskNetwork}[1][none]{\possibleBold{#1}{\mathcal{T}}}
\newcommand{\robotTraitMatrix}[1][none]{\possibleBold{#1}{Q}}

\newcommand{\numRobots}[1][none]{\possibleBold{#1}{N}}
\newcommand{\numTasks}[1][none]{\possibleBold{#1}{M}}

\newcommand{\problemDomainStatic}[1][none]{\possibleBold{#1}{\mathcal{D}}}

\newcommand{\solutionStatic}[1][none]{\possibleBold{#1}{S}}

\newcommand{\setMotionPlans}[1][none]{\possibleBold{#1}{\ensuremath{X}}}

\newcommand{\allocationOptimal}[0]{\ensuremath{\allocation^*}}

\newcommand{\algoname}[0]{Q-ITAGS}

\begin{document}
\mainmatter              
\title{{Q-ITAGS}: Quality-Optimized Spatio-Temporal Heterogeneous Task Allocation with a Time Budget
\vspace{-8pt}
}
\titlerunning{Q-ITAGS}  
%
\author{Glen Neville\inst{1}$^\dagger$ \and Jiazhen Liu\inst{1}$^\dagger$ \and
Sonia Chernova\inst{1} \and Harish Ravichandar\inst{1} \\ $^\dagger$ Equal contribution}
\authorrunning{Glen Neville et al.} 
%
%
\tocauthor{Glen Neville, Jiazhen Liu, Sonia Chernova, Harish Ravichandar}
\institute{Georgia Institute of Technology, Atlanta GA 30332, USA,\\
\email{gneville, jliu3103, Chernova, harish.ravichandar@gatech.edu}
}

\maketitle              

\vspace{-10pt}
\begin{abstract}
Complex multi-objective missions require the coordination of heterogeneous robots at multiple inter-connected levels, such as coalition formation, scheduling, and motion planning. 
The associated challenges are exacerbated when solutions to these interconnected problems need to simultaneously maximize task performance and respect practical constraints on time and resources. 
In this work, we formulate a new class of spatio-temporal heterogeneous task allocation problems that formalize these complexities.
We then contribute a novel framework, named Quality-Optimized Incremental Task Allocation Graph Search (\algoname), to solve such problems. 
\algoname\ 
offers a flexible \textit{interleaved} framework that 
i) explicitly models and optimizes the effect of collective capabilities on task performance via learnable \textit{trait-quality maps}, and
ii) respects both resource and spatio-temporal constraints, including a user-specified \textit{time budget} (i.e., maximum makespan).
In addition to algorithmic contributions, we derive theoretical \textit{suboptimality bounds} in terms of task performance that varies as a function of a single hyperparameter. 
Detailed experiments involving a simulated emergency response task and a real-world video game dataset reveal that i) \algoname\ results in superior team performance compared to a state-of-the-art method, while also respecting complex spatio-temporal and resource constraints, ii) \algoname\ efficiently learns trait-quality maps to enable effective trade-off between task performance and resource constraints, and iii) \algoname' suboptimality bounds consistently hold in practice. 


\end{abstract}

\vspace{-8pt}
\section{Introduction}

Heterogeneous multi-robot systems (MRS) bring together complementary capabilities to tackle challenges in domains as diverse as agriculture \cite{Roldan2016}, defense \cite{McCook2007}, assembly \cite{Stroupe2005}, and warehouse automation \cite{Baras2019}.  Achieving effective teaming in such complex domains requires reasoning about task allocation (\textit{who})~\cite{Ravichandar2020}, scheduling (\textit{when})~\cite{Matos2021}, motion planning (\textit{how})~\cite{Baras2019}, and their intersections~\cite{Neville2021,Messing2022}. 

In this work, we tackle the challenge of \textit{spatio-temporal task allocation} for heterogeneous teams of robots, which requires tackling task allocation, scheduling, and motion planning \textit{simultaneously}. Specifically, we focus on relaxing three key limiting assumptions of related methods in the current literature. 

First, several methods assume that one can either \textit{decompose multi-robot tasks} into individual robot subtasks~\cite{Giordani2013,Krizmancic2020} or \textit{enumerate} all possible coalitions~\cite{gombolay2016,Schillinger2018}. However, it is often difficult, if not impossible, to explicitly specify the role of each robot in tasks that involve complex collaboration, and enumerating every capable coalition might be intractable. 
Second, existing methods tend to assume a \textit{binary success-or-failure model} of task outcomes~\cite{fu2022robust,Messing2022,srikanthan2022resource}. Such a limited model does not apply to several real-world problems which require maximization of performance metrics, as opposed to satisfaction of arbitrary thresholds (e.g., distributed sensing, supply chain optimization, and disaster response).
Third, existing methods often demand that users \textit{explicitly specify} the requirements for successful task completion~\cite{Prorok2017,Ravichandar2020,Mayya2021}. However, it is well known that humans, while adept at making complex decisions, often fail to precisely articulate how they do so; in fact, we tend to make things up to justify our decisions~\cite{rieskamp2003people,rieskamp2006ssl}. 
We refer the reader to Sec.~\ref{sec:related_work} for a detailed discussion of related works.

\begin{figure}[t]
\begin{center}
    \includegraphics[width=0.8\columnwidth]{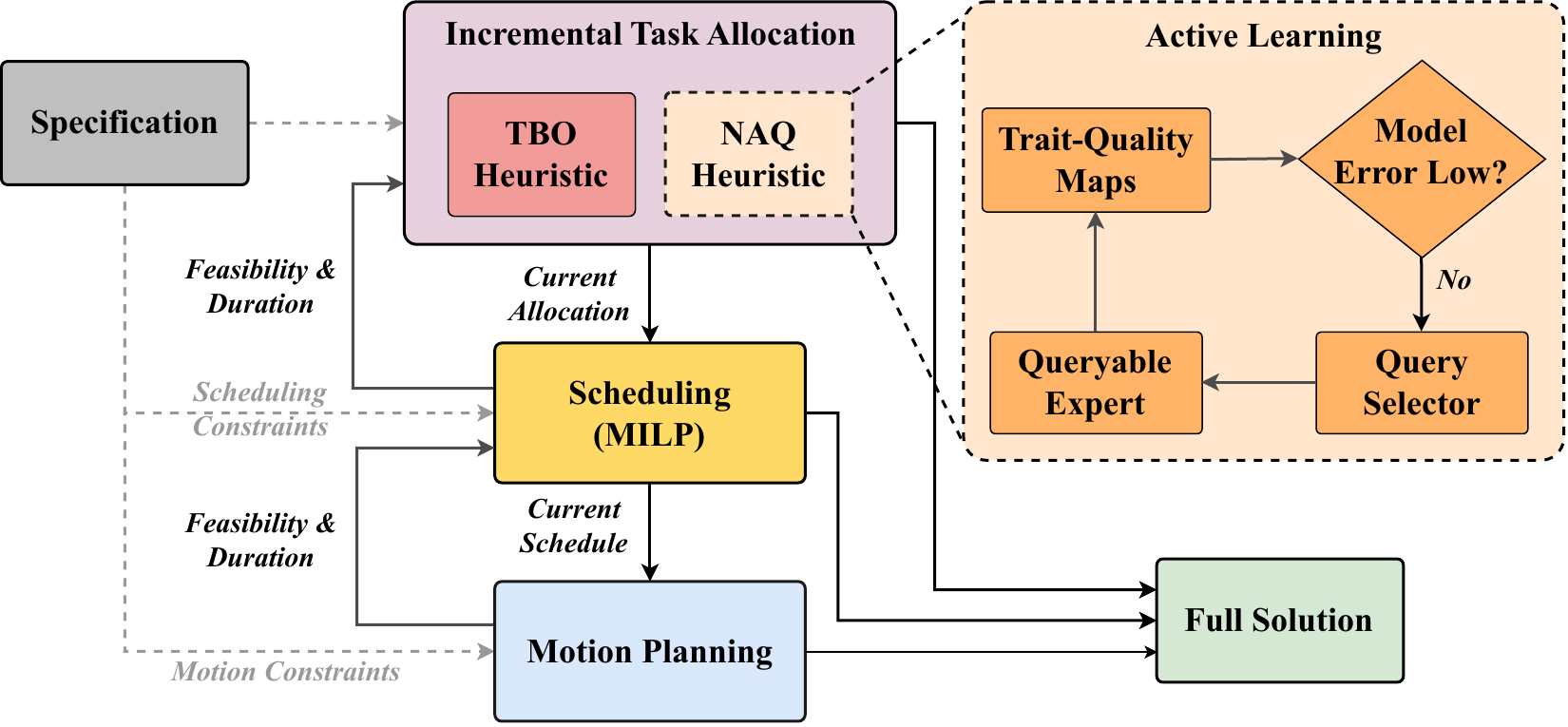}
\end{center}
\caption{
\small{
\algoname{} performs spatio-temporal task allocation for heterogeneous multi-robot teams by optimizing collective performance while respecting spatio-temporal and resource constraints. It explicitly models, actively learns, and optimizes \textit{trait-quality maps} that approximate the effects of collective capabilities on task performance.
}}
\label{fig:block_diagram} 
\end{figure}

We contribute \textit{Quality-optimized Incremental Task Allocation Graph Search (\algoname)}\footnote{Open-source implementation: \url{https://github.com/GT-STAR-Lab/Q-ITAGS}}, 
a heuristic-driven interleaved search algorithm that sheds critical assumptions of existing methods and overcomes the corresponding limitations (see Fig.~\ref{fig:block_diagram}). 
First, improving upon our prior work in \textit{trait-based} task allocation~\cite{Ravichandar2020,Neville2021,Messing2022,srikanthan2022resource,neville-2023-d-itags}, \algoname{} models both agents and tasks in terms of capabilities and task requirements in terms of collective capabilities, leading to a \textit{flexible} and \textit{robot-agnostic} framework. As such, \algoname{} neither requires decomposition of multi-robot tasks nor the enumeration of all possible coalitions.
Second, \algoname{} utilizes \textit{continuous trait-quality maps}, a novel and expressive model of performance that maps collective capabilities to task performance, helping encode and optimize the quality of allocations.
Third, \algoname{} does not demand explicit specification of task requirements. Instead, it employs an \textit{active learning} module to learn how collective capabilities relate to task performance.


We formalize a new class of problems, dubbed \textit{spatio-temporal task allocation with quality optimization (STAQ)}, which requires maximizing the quality of allocations (i.e., task performance) while simultaneously ensuring that the makespan of the associated schedule is less than a user specified threshold (i.e., time budget).
To solve STAQ problems, we then develop two heuristics that guide the search within \algoname{}. Specifically, we develop i) \textit{Time Budget Overrun (TBO)}, which captures the suboptimality of a given schedule anchored against a user-specified threshold on makespan, and ii) \textit{Normalized Allocation Quality (NAQ)}, which provides a normalized estimate of the current allocation's quality. 

We extensively evaluated \algoname{} in terms of its ability to i) optimize task allocation quality, ii) respect time budgets, and iii) learn unknown mappings between collective capabilities and task performance. We conducted our evaluations on both a simulated emergency response domain and a publicly-available dataset of a multi-player game. Our results conclusively demonstrate that \algoname{} outperforms the state-of-the-art in terms of allocation quality while generating schedules of a similar makespan. Further, \algoname' active sampling approach learns unknown task quality mappings consistently more efficiently than a passive approach that learns from uniformly sampled data.


In addition to empirical benefits, we provide theoretical insights into \algoname' operation. We derive a bound on \algoname' sub-optimality in terms of allocation quality under mild assumptions. 
Notably, our analysis illuminates an inherent trade-off between allocation quality and schedule makespan that hinges on a single hyperparameter. This insight and analysis will provide users with intuitive guidance about the trade-off when choosing the hyperparameter.
We also demonstrate via experiments that our bounds consistently hold in practice.

In summary, we contribute: i) a formalism for a new class of problems in spatio-temporal task allocation that focus on optimizing allocation quality while respecting makespan constraints, ii) an interleaved search algorithm and associated heuristics to effectively solve such problems, iii) an active learning method to learn unknown relationships between collective robot capabilities and task performance, and iv) theoretical bounds on suboptimality in allocation quality.


\section{Related Works}
\label{sec:related_work}

While the multi-robot task allocation problem has many variants~\cite{Gerkey2004,Nunes2017}, we limit our discussion to one variant to which problem is closely related: single-task (ST) robots, multi-robot (MR) tasks, and time-extended (TA) allocation. ST-MR-TA problems require assigning coalitions to tasks while respecting temporal constraints (e.g., precedence and ordering), spatio-temporal constraints (e.g., travel time), and deadlines. In addition to these issues, our problem formulation also involves spatial constraints (e.g., task locations and obstacles).

\textit{Auction-based methods} to solve ST-MR-TA problems involve auctioning tasks to robots through a bidding process based on a utility function that combines the robot's (or the coalition's) ability to perform the task with any temporal constraints \cite{Giordani2013,Krizmancic2020,quinton2023market}. Auctions are highly effective, but typically require each robot to know its individual utility towards a task. As a result, these methods require either 
i) multi-robot tasks to be decomposable into sub-tasks, each solvable by a single robot, 
or ii) robots within a coalition to have explicit knowledge of their individual contribution to all collaborative tasks.

\textit{Optimization-based methods} formulate the ST-MR-TA problem as a mixed-integer linear program (MILP) to optimize the overall makespan or a utility function~\cite{Korsah2012,Guerrero2017,chakraa2023optimization}.
However, these methods often require that all tasks be decomposable into single-agent tasks \cite{Schillinger2018}. In contrast to auction-based and optimization-based approaches, our approach does not require task decomposition.

Our approach to task allocation is most closely related to \textit{trait-based methods}~\cite{Koes2006,Prorok2017,Ravichandar2020,Neville2020,Neville2021,srikanthan2022resource,fu2022robust,fu2022learning,Singh-RSS-23,neville-2023-d-itags,10341837}, which utilizes a flexible modeling framework that encodes task requirements in terms of traits (e.g., Task 1 requires traveling at 10m/s while transporting a 50-lb payload). Each task is not limited to a specific set or number of agents. Instead, the focus is on finding a coalition of agents that \textit{collectively} possess the required capabilities. As such, these methods allow for flexible and automated determination of coalition sizes and compositions.  However, most existing trait-based approaches are limited to single-task robot, multi-robot task, instantaneous allocation (ST-MR-IA) problems that do not require scheduling~\cite{Koes2006,Prorok2017,Ravichandar2020,Neville2020,srikanthan2022resource,Singh-RSS-23}, with a few notable exceptions that can handle ST-MR-TA problems~\cite{Neville2021,neville-2023-d-itags,fu2022robust,fu2022learning}.  


A key limitation of existing trait-based algorithms, including~\cite{Neville2021,neville-2023-d-itags,fu2022robust,10341837}, is that they assume that the user will explicitly quantify the minimum amount of collective traits necessary to perform each task. 
This presents two issues.
First, explicit quantification of minimum requirements can be intractable even for experts when dealing with complex tasks and multi-dimensional capabilities~\cite{rieskamp2003people}. 
Other sub-fields within robotics recognize this concern and have developed methods to prevent and handle reward or utility mis-specification~\cite{Menell2017,Mallozzi2018}.
Second, all methods, including those that learn requirements~\cite{fu2022learning} assume that any and all additional traits allocated beyond the minimum provide no benefit to the team and that not satisfying the specified threshold will lead to complete failure. This effectively ignores the natural variation in task performance as a function of allocated capabilities.
In stark contrast, \algoname{} does not require the user to specify trait requirements and utilizes a more expressive model to capture the relationship between allocated capabilities and task performance.

\section{Problem Formulation}
\label{sec:problemdef}

We begin with preliminaries from our prior work~\cite{Neville2021,Messing2022,neville-2023-d-itags} for context, and then formalize a new class of spatio-temporal task allocation problems.

\subsection{Preliminaries}


Consider a team of \numRobots\ heterogeneous robots, with the $i$th robot's capabilities described by a collection of traits $q^{(i)} = \sqBrackets{q_1^{(i)},\ q_2^{(i)},\ \cdots, q_U^{(i)}}^\intercal \in \mathbb{R}_{\geq 0}^U$, where $q_u^{(i)}$ corresponds to the $u^{th}$ trait for the $i^{th}$ robot. We assign $q_u^{(i)} = 0$ when the $i^{th}$ robot does not possess the $u^{th}$ trait (e.g. firetrucks have a water capacity, but other robots may not).
As such, the capabilities of the team can be defined by a \textbf{team trait matrix}:
$$
  \robotTraitMatrix[bold]  = \sqBrackets{q^{(1)},\ \cdots,\ q^{(N)}}^{\intercal} \in \R_{\geq 0}^{N \times U}
$$ 
whose $iu$th element denotes the $u^{th}$ trait of the $i^{th}$ robot.

We model the set of $M$ tasks to be completed as a \textbf{Task Network} \taskNetwork[bold]: a directed graph $G=(\mathcal{V}, \mathcal{E})$, with vertices $\mathcal{V}$ representing a set of tasks $\{\task_m\}_{m=1}^{M}$, and edges $\mathcal{E}$ representing constraints between tasks. For a pair of tasks $\task_i$ and $\task_j$ with $a_i \neq a_j$, we consider two kinds of constraints: (i) \textbf{precedence constraint} $\task_i \prec \task_j$ requires that Task $\task_i$ should be completed before the Task $\task_j$ can begin (e.g., a fire must be put out before repairs can begin) \cite{Weld1994}; (ii) \textbf{mutex constraint} ensures that $\task_i$ and $ \task_j$ are not concurrent (e.g. a manipulator cannot pick up two objects simultaneously) \cite{Bhargava2019}. Further, we define \maxMakespan~as the \textbf{total time budget} which encodes the longest makespan acceptable to the user.

Let $\allocation \in \allocationSpace \subseteq \{0,1\}^{M \times N}$ denote the binary \textbf{allocation matrix}, where $\allocation_{mn} = 1$ if and only if the $n$th robot is allocated to the $m$th task and $\allocation_{mn} = 0$ otherwise. Robots can complete tasks individually or collaborate as a coalition, and any robot can be allocated to more than one task as long as the tasks are scheduled to be carried out during non-overlapping intervals. We use $\allocation_m$ to denote the $m$th row of $\allocation$ and it specifies the coalition of robots assigned to the $m$th task. We further define 
$\mathbf{Y} = \allocation \robotTraitMatrix[bold] \in \R^{M\times U}_{\geq 0}$ to be the aggregated traits for all tasks according to allocation plan $\allocation $, with $y_m$ denoting its $m$th row containing the collective traits available for the $m$th task. 

\subsection{Quality-Optimized Spatio-Temporal Task Allocation}
In this work, we extend the problem formulation from our prior work to account for the quality of task allocation. We formulate a new class of problems, \textit{spatio-temporal task allocation with quality optimization (STAQ)}, which involves optimizing the quality of robots' assignments to tasks while ensuring that the associated makespan is less than a user-specified threshold. We then extend STAQ to include \textit{active learning} of unknown trait-quality maps, which map the collective capabilities of a given coalition to task performance.

Let $\performanceFunction_m: \R^{U}_{\geq 0} \rightarrow [0,1]$ be the normalized \textbf{trait-quality map} that computes a non-negative quality score (with numbers closer to 1 indicating higher quality coalitions) associated with the $m^{th}$ task given the collective traits allocated to it. In essence, $\performanceFunction_m$ quantifies how a given coalition's capabilities translate to performance on the $m^{th}$ task. We define the \textbf{total allocation quality} to summarize the performance of all tasks:
\begin{equation}
    \performanceFunctionSet(\allocation) = \sum_{m=1}^{M} \performanceFunction_m(\robotTraitMatrix[bold]^T\allocation_m^T) = \sum_{m=1}^{M}  \performanceFunction_m(y_m) 
    \label{eq:total_allocation_quality}
\end{equation}
Note that the maps $\performanceFunction_m, \forall m$ are between \textit{traits} (not robots) and performance, yielding more generalizable relationships. We also discuss the active learning of trait-quality maps in Sec.~\ref{sec:active_learning_staq} when they are unknown. 


We define the \textbf{problem domain} using the tuple $\problemDomainStatic[bold] = \big<\taskNetwork[bold],\ \robotTraitMatrix[bold],\ \performanceFunctionSet,\ \setInitialConfigurations,\ \setInitialTerminalConfigurationsTaskNetwork, \\ \worldState{}, \maxMakespan \big>$,
where
\taskNetwork[bold]\ is the task network, \robotTraitMatrix[bold]\ is the team trait matrix, \performanceFunctionSet\ is the summarized performance function, $\setInitialConfigurations$ and $\setInitialTerminalConfigurationsTaskNetwork$ are respectively the sets of all initial and terminal configurations associated with tasks, \worldState{} is a description of the world (e.g., a map), and \maxMakespan~is the total time budget.

Finally, we define a \textbf{solution} to the problem defined by $\problemDomainStatic[bold]$ using the tuple $\boldsymbol{\solutionStatic} = \angled{\allocationSolution,\ \scheduleSolution,\ \mpSolution}$, where $\allocationSolution$\ is an allocation, $\scheduleSolution$\ is a schedule that respects all temporal constraints, and $\mpSolution$\ is a finite set of collision-free motion plans. The goal of STAQ is to find a solution that maximizes the total allocation quality $\performanceFunctionSet(\allocation)$ while respecting the total time budget $\maxMakespan$.



\subsection{Active Learning for STAQ}
\label{sec:active_learning_staq}
As above, $\performanceFunction_m(y_m)$ is a function that maps the collective traits assigned to the $m$th task to a measure of task performance. While accurate knowledge of these maps is needed to make effective trade-offs when allocating limited resources (i.e., robots) to competing objectives (i.e., tasks), they are often difficult to explicitly specify.
As such, we do not assume knowledge of $\performanceFunction_m(\cdot),\ \forall m$. Instead, we assume access to a simulator from which to learn. We are specifically interested in limiting the number of queries made to facilitate learning since such queries can be expensive or time-consuming. As such, we formulate the problem of learning $\performanceFunction_m(\cdot),\ \forall m$ as an active learning problem in which one must effectively sample the most informative coalitions to efficiently learn the trait-quality maps. 


\vspace{-0.2cm}
\section{Approach}

To solve the STAQ problems as defined in Sec. \ref{sec:problemdef}, we propose \textbf{Quality-optimized Iterative Task Allocation Graph Search (\algoname)} algorithm. We begin with a brief overview of \algoname{} before supplying details. 

\algoname{} utilizes an interleaved architecture to simultaneously solve task allocation, scheduling, and motion planning (see Fig.~\ref{fig:block_diagram}). This interleaved approach is inherited from our prior work~\cite{Neville2021,Messing2022,neville-2023-d-itags}, in which we provide extensive evidence for its benefits over sequential counterparts that are commonly found in the literature. For instance, interleaved architectures are considerably more computationally efficient since they avoid backtracking~\cite{Messing2022}.
However, \algoname{} makes crucial improvements over our prior work: i) optimizing for allocation quality instead of makespan, ii) utilizing a more expressive model of task performance, and iii) respecting a time budget (see Sec.~\ref{sec:related_work}).

\algoname{} performs an incremental graph-based search that is guided by a novel heuristic. Our heuristic balances the search between optimizing task allocation quality and meeting the time budget requirement. We formulate and solve a mixed-integer linear program (MILP) to address the scheduling and motion planning problems as part of the incremental search. To get around the challenge of explicitly specifying task requirements, \algoname{} employs an \textit{active learning} module to learn the relationship between the collective traits of the coalition and task performance. 
When a very tight time budget $C_{max}$ is imposed, it might be possible that no feasible solution exists no matter how the robots are allocated. For such cases, \algoname{} can signal infeasibility after it exhausts all combinations through the search without finding a solution.

\subsection{Task Allocation}\label{subsec:allocation} 
The task allocation layer performs a greedy best-first search through an incremental task allocation graph $\mathcal{N}$. In this graph, each node represents an allocation of robots to tasks. Nodes are connected to other nodes that differ only by the assignment of a single robot. 
Note that the root node represents all robots allocated to each task.  Indeed, such an allocation will result in the best performance (since all available robots contribute to each task), but will result in the longest possible schedule (sequential execution of tasks). \algoname' search starts from this initial node, and incrementally removes assignments to find an allocation that can meet the time budget without significantly sacrificing task performance.

To guide the search, we developed two heuristics: \textit{Normalized Allocation Quality (NAQ)}, which guides the search based on the quality of the allocation, and \textit{Time Budget Overrun (TBO)}, which guides the search based on the makespan of the schedule associated with the allocation. 

\noindent \textbf{\textit{Normalized Allocation Quality (NAQ)}} is a normalized measure of how a given allocation improves task performance, and is calculated as
\begin{equation}
    \label{equ:naq}
    f_{NAQ}(\overline{\allocation})  = \frac{\performanceFunctionSet(\allocation_\text{root}) - \performanceFunctionSet(\overline{\allocation})}{\performanceFunctionSet(\allocation_\text{root}) - \performanceFunctionSet(\allocation_\text{null})}    
\end{equation}
where $\overline{\allocation}$ is the allocation being evaluated, $\performanceFunctionSet(\overline{\allocation})$ is the total allocation quality as defined in Eq.~(\ref{eq:total_allocation_quality}), $\allocation_\text{root}$ is the allocation at the root node, and $\allocation_\text{null}$ represents no robots being assigned to any task. Since $\performanceFunctionSet(\allocation_\text{root})$ and $\performanceFunctionSet(\allocation_\text{null})$ respectively define the upper and lower bounds of allocation quality, NAQ is bounded within $[0,1]$. Note that the quality functions $\{\performanceFunction_m(\cdot)\}_{m=1}^{M}$, needed to compute $\performanceFunctionSet(\cdot)$, would be learned as described in Sec.~\ref{subsec:active_learning}.

NAQ does not consider the scheduling layer and promotes a broader search of the allocation graph. This is due to the fact that shallower nodes have more robots allocated and as a result tend to result in higher allocation quality. As such, NAQ favors allocations that maximize allocation quality and task performance at the expense of a potentially longer makespan.

\vspace{3pt}
\noindent \textbf{\textit{Time Budget Overrun (TBO)}} measures how much the schedule associated with a given allocation violates the time budget, and is calculated as 
\begin{equation}
    \label{equ:pos}
    f_{TBO}(C_{\overline{\schedule[bold]}}) = \max \left( \frac{C_{\overline{\schedule[bold]}} -  C_{max}}{\vert C_{\schedule[bold]_{worst}} -  C_{max}\vert},\ 0 \right)    
\end{equation}
where $C_{\schedule[bold]}$ denotes the makespan of the schedule $\schedule[bold]$, $\overline{\schedule[bold]}$ is the schedule associated with the $\overline{\allocation}$ being evaluated, $C_{max}$ is the user-specified time budget for makespan, and $\schedule[bold]_{worst}$ is the longest schedule which allocates all robots to each task.

Since TBO only considers the schedule and not the allocation, it tends to favor nodes deeper in the graph that have fewer agents and constraints and thus a lower makespan. 
As such, TBO favors allocations that satisfy the time budget at the expense of a deeper search and more node expansions.

\vspace{3pt}
\noindent \textbf{\textit{Time-Extended Task Allocation Metric (TETAM)}} is a convex combination of NAQ and TBO, balancing allocation quality and time budget:
\begin{equation}
    \label{equ:TETAM}
    f_{TETAM}(\overline{\allocation}, C_{\overline{\schedule[bold]}})= (1-\alpha)f_{NAQ}(\overline{\allocation})  + \alpha f_{TBO}(C_{\overline{\schedule[bold]}})
\end{equation}
where $\alpha \in [0,1]$ is a user-specified parameter that controls each heuristic's relative influence. 
Thus, TETAM considers both allocation quality and the schedule simultaneously when searching for a solution. See Sec.~\ref{sec:theorems} for a theoretical analysis of the trade-offs between allocation quality and makespan.

\subsection{Scheduling and Motion Planning} \label{subsec:scheduling_and_MP} 
\algoname{}' scheduling layer checks the feasibility of scheduling a given allocation and helps compute its TBO. To this end, we formulate a mixed-integer linear program (MILP) that considers three types of temporal constraints: precedence, mutex, and travel time. Precedence constraints $\mathcal{P}$ ensure that one task happens before another (e.g., the fire must be doused before repairs begin). Mutex constraints $\mathcal{M}$ ensure that two tasks do not happen simultaneously (e.g., a robot cannot pick up two objects simultaneously). Travel time constraints ensure that robots have sufficient time to travel between task sites (e.g., traveling to the location of fire before dousing). The MILP takes the following form:

\begin{equation*}
\begin{aligned}
    \min_{\{s_i\}, \{ p_{ij}\}}  &\ C \\
    \mathrm{s.t.}    &\ C \geq s_i + d_i,\ \forall i=1,..,\numTasks \\
                     &\ s_j \geq s_i + d_i + x_{ij},\ \forall (i,j) \in \mathcal{P} \\
                     &\ s_i \geq x_i,\ \forall i=1,..,\numTasks \\
                     &\ s_i \geq s_j + d_j + x_{ji} - \beta p_{ij},\ \forall (i,j) \in \mathcal{M}^R \\
                     &\ s_j \geq s_i + d_i + x_{ij} - \beta(1-p_{ij}),\ \forall (i,j) \in \mathcal{M}^R
\end{aligned}    
\end{equation*}
where $C$ is the makespan, $s_i$ and $d_i$ are the (relative) starting time and duration of Task $\task_i$, $x_{ij}$ is the time required to translate from the site of $\task_i$ to the site of $\task_j$, $x_i$ is the initial time required for the allocated coalition to travel to the site of Task $a_i$, $p_{ij} \in \{ 0, 1\}$, and $p_{ij} = 1$ if and only if $\task_i$ precedes $\task_j$, $\beta \in \mathbb{R}_+$ is a large scalar, and $\mathcal{P}$ and $\mathcal{M}$ are sets of integer pairs containing the lists of precedence and mutex constraints, with $\mathcal{M}^R = \mathcal{M} - \mathcal{P} \cap \mathcal{M}$ denoting mutex constraints with precedence constraints removed.

\algoname{} initializes the schedule by estimating travel times using Euclidean distances. During the search, the scheduling layer iteratively queries the motion planner to better estimate travel times. \algoname{} iterates until all motion plans required by the schedule are instantiated and memoized.

\subsection{Active Learning}\label{subsec:active_learning} 

Below, we discuss how \algoname{} models the trait-quality maps $\{\performanceFunction_m(\cdot)\}_{m=1}^{M}$, and learns them from actively sampled data. To facilitate learning, we utilize a simulator that can estimate the task performance of any given coalition. 
Note that when a simulator is not available, a human expert can be queried for labels. 
\algoname' active learning module has two key components: 

\vspace{4pt}
\noindent \textit{Trait-Quality Maps}:
As discussed above, trait-quality maps specify the relationship between collective capabilities and task performance. Formally, we model the joint distribution of collective traits assigned to a task and the resulting task performance as a Gaussian Process with a radial basis function kernel: $P\left(y_m,\performanceFunction_m(y_m)\right) \sim \mathcal{GP}(\mu_m,\Sigma_m), \forall m$. As such, we can predict the quality of a given collection of traits using Gaussian Process Regression. Note that this design choice allows \algoname{} to not just estimate quality, but also to quantify its uncertainty, which in turn helps in active learning. 


\vspace{4pt}
\noindent \textit{Query Selector}:
To ensure efficient learning, \algoname' query selector samples the most-informative data points from a collection of unlabeled data points (i.e., coalitions that have not been sampled). \algoname{} employs a maximum-entropy sampling approach that computes the uncertainty of the current trait-quality map's quality estimate for each unlabeled data point, and then selects the data point where the uncertainty is the highest. Note that we do not use upper confidence bound (UCB) sampling since, during learning, we are more interested in coverage rather than in optimizing task performance. As such, our maximum entropy approach chooses samples in regions where the current model is the most ignorant, building a more comprehensive model that will inform trade-off decisions when allocating limited resources (i.e., robots) to competing tasks. Further, experiments reveal that our active learning approach results in better sample efficiency and more consistent learning, compared to a passive approach which employs uniform sampling (see Sec. \ref{subsec:nfl_madden}).


\section{Bounds on Suboptimality}
\label{sec:theorems}

To better understand \algoname' performance, we analyze the effect of $\alpha$ -- the user-specified parameter that determines the relative importance of our two heuristics -- on the optimality of the obtained solution defined with respect to the quality of task allocation. We demonstrate that the choice of $\alpha$ determines a suboptimality bound, where $\alpha$ values between $\alpha=0$ and $\alpha=0.5$ promise increasingly tighter bounds on suboptimality. This analysis and the fact that $\alpha$ values closer to 1 provide solutions that prioritize makespan can offer an intuitive understanding of the relationship between allocation quality and makespan.

Below, we derive strict suboptimality bounds for solutions generated by \algoname{}, in terms of total allocation quality.

\begin{theorem}\label{thm:makespan}
    For a given problem domain \problemDomainStatic[bold], let $\allocationOptimal$ be the optimal allocation w.r.t. total allocation quality under a time budget, and $\allocationSolution$ be the allocation of the solution generated by \algoname. If $\alpha < 0.5$ in Eq. (\ref{equ:TETAM}), then
    \begin{equation}
    \label{equ:boundsPre}
         \performanceFunctionSet(\allocationOptimal)- \performanceFunctionSet(\allocationSolution) \leq \frac{\alpha}{1-\alpha} (\performanceFunctionSet(\allocation_\text{root})- \performanceFunctionSet(\allocation_\text{null})),
    \end{equation}
    where $\performanceFunctionSet(\allocation_\text{root})$ and $\performanceFunctionSet(\allocation_\text{null})$ denote the total allocation quality respectively when all robots are assigned to all tasks, and when no robot is assigned to any task.
    \label{thm1}
\end{theorem}

\begin{proof}

Since any expansion of a parent node represents the subtraction of an assignment, any given node $N$ is guaranteed to have fewer agents assigned than its parent $N_p$. This observation, when combined with the fact that removing agents can never improve the allocation quality (as removing an agent always decreases the traits of the coalition), yields

 \begin{equation}
 \label{equ:monoNSQ}
    f_{NAQ}(N) \geq f_{NAQ}(N_p),
 \end{equation}
since a higher allocation quality results in a lower NAQ heuristic value and vice versa. Consequently, we can infer that the NAQ value of all nodes in the unopened set $\mathcal{U} \subseteq \mathcal{N}$ of the \algoname~graph is greater than or equal to that of their respective predecessors in the opened set $\mathcal{O} \subseteq \mathcal{N}$.
As such, the smallest NAQ value in the unopened set must be greater than that in the opened set:

 \begin{equation}
 \label{equ:schedIncrease}
     \min_{N \in \mathcal{U}}f_{NAQ}(N) \geq \min_{N \in \mathcal{O}}f_{NAQ}(N)
 \end{equation}

Irrespective of whether the optimal allocation $\allocationOptimal$ corresponds to a node in the open or unopened set, 
the inequality in Eq.~(\ref{equ:schedIncrease}) implies that
 \begin{equation}
 \label{equ:boundOpt}
\performanceFunctionSet(\allocationOptimal) \leq \max_{N \in \mathcal{O}} \performanceFunctionSet(\allocation_N)
 \end{equation}
where $\allocation_N$ denotes the allocation of a given node $N$. We use $N'$ to denote the node from the open set with the maximum total allocation quality. Namely, $N' = \argmax_{N \in \mathcal{O}}\performanceFunctionSet(\allocation_N)$.
 
Since we require any valid solution of \algoname{} to respect the time budget, the solution node ($\solutionNode$) will have a zero TBO value. 
As such, the total TETAM heuristic defined in Eq.~(\ref{equ:TETAM}) for the solution node is given by
\begin{equation}
 \label{equ:tetamsol}
      f_{TETAM}(\solutionNode) = (1- \alpha) f_{NAQ}(\solutionNode) 
 \end{equation}
Since \algoname{} performs greedy best-first search, the TETAM value of the solution node is guaranteed to be less than or equal to all nodes in the open set:
 \begin{equation}
 \label{equ:soldef_queue}
    f_{TETAM}(\solutionNode) \leq f_{TETAM}(N), \\ \forall N \in \mathcal{O} 
 \end{equation}
Expanding the definition of TETAM and using Eq.~(\ref{equ:tetamsol}) yields
 \begin{equation}
 \label{equ:simplify}
    \begin{aligned}
    & (1-\alpha)\frac{ \performanceFunctionSet(\allocation_\text{root}) - \performanceFunctionSet(\allocationSolution)}
    {\performanceFunctionSet(\allocation_\text{root})- \performanceFunctionSet(\allocation_\text{null})}
    \leq \\
    & \alpha  f_{TBO}(N) + (1-\alpha)\frac{\performanceFunctionSet(\allocation_\text{root}) - \performanceFunctionSet(\allocation_{N})}
    {\performanceFunctionSet(\allocation_\text{root}) - \performanceFunctionSet(\allocation_\text{null})}, \forall N \in \mathcal{O}  
    \end{aligned}
 \end{equation}
Using the bound in Eq.~(\ref{equ:boundOpt}), and the fact that $ f_{TBO}(\cdot) \leq 1$, we rewrite the RHS of above equation as
 \begin{equation}
    \label{equ:simplify_2}
    \begin{aligned}
            & (1-\alpha)\frac{\performanceFunctionSet(\allocation_\text{root}) -  \performanceFunctionSet(\allocationSolution)}
            {\performanceFunctionSet(\allocation_\text{root}) - \performanceFunctionSet(\allocation_\text{null})}
            \leq  \\
            & \alpha + (1-\alpha)\frac{\performanceFunctionSet(\allocation_\text{root}) - \performanceFunctionSet(\allocationOptimal)}
            {\performanceFunctionSet(\allocation_\text{root}) - \performanceFunctionSet(\allocation_\text{null})
            },  \forall N \in \mathcal{O},  
    \end{aligned}
 \end{equation}
After rearranging and canceling equivalent terms on both sides, we get
 \begin{equation}
 \label{equ:simplify_3}
 \begin{aligned}
    & (1-\alpha)(\performanceFunctionSet(\allocation_\text{root}) - \performanceFunctionSet(\allocationSolution))
    \leq \\
    & \alpha (\performanceFunctionSet(\allocation_\text{root}) - \performanceFunctionSet(\allocation_\text{null})) + (1-\alpha) (\performanceFunctionSet(\allocation_\text{root})) - \performanceFunctionSet(\allocationOptimal))
\end{aligned}
 \end{equation} 
Finally, rearranging the terms yields the bound in Eq.~(\ref{equ:boundsPre}).
\end{proof}

Note that the above theorem provides sensible bounds on the difference in allocation quality between the optimal solution and \algoname' solution only when $0 \leq \alpha < 0.5$. When $\alpha \geq 0.5$, the bound loses significance as it grows beyond the maximum difference in quality between $\performanceFunctionSet(\allocation_\text{root})$ and $\performanceFunctionSet(\allocation_\text{null})$. In other words, Eq.~\ref{equ:boundsPre} holds trivially when $\alpha \geq 0.5$, since $\performanceFunctionSet(\allocation_\text{null})$ has the worst quality metric of all allocations and $\performanceFunctionSet(\allocation_\text{root})$ has the best one, meaning a bound greater than $\performanceFunctionSet(\allocation_\text{root}) - \performanceFunctionSet(\allocation_\text{null})$ would not be significant.

The bound in Eq.~(\ref{equ:boundsPre}) can be tightened after the execution to facilitate post-hoc analyses. Specifically, instead of bounding $f_{TBO}(N), \forall N \in \mathcal{O}$ by 1, 
we can exactly compute $f_{TBO}(N')$ for $N'$, which is the node yielding the best allocation quality.
Then, following similar algebraic manipulations as above yields a tighter bound on the suboptimality gap:
\begin{equation*}
    \label{equ:boundsPost}
     \performanceFunctionSet(\allocationOptimal)- \performanceFunctionSet(\allocationSolution) \leq \frac{\alpha}{1-\alpha} (\performanceFunctionSet(\allocation_\text{root}) - \performanceFunctionSet(\allocation_\text{null})) f_{TBO}(N’)
\end{equation*}

\section{Evaluation}
\label{sec:evalutation}

We evaluated \algoname{} using three sets of experiments. The first set of experiments evaluated \algoname{} in a simulated emergency response domain~\cite{Kitano1999,Bechon2014,Messing2020,Whitbrook2015,Zhao2016} and compared its performance against a state-of-the-art trait-based task allocation method. The second set of experiments examined if the theoretical bounds on resource allocation suboptimality held in practice. In the third set of experiments, we evaluated the active learning strategy of \algoname{} on a real-world dataset involving the American Football video game NFL Madden. 

\subsection{Coordinated Emergency Response}
\label{subsec:ITAGS_comparison}

We evaluate \algoname' performance in a simulated emergency response domain. Robots with heterogeneous traits are required to collaborate on tasks such as extinguishing fire and rescuing wounded survivors (see \cite{Neville2021,Messing2022} for more details). To accommodate comparisons against a SOTA trait-based task allocation method, we assume that the trait quality maps are known and represented by a linear weighted sum of coalitions' collective traits.

\begin{figure}[t]
\begin{center}
    \includegraphics[width=0.65\columnwidth, trim=0 0 0 0, clip]{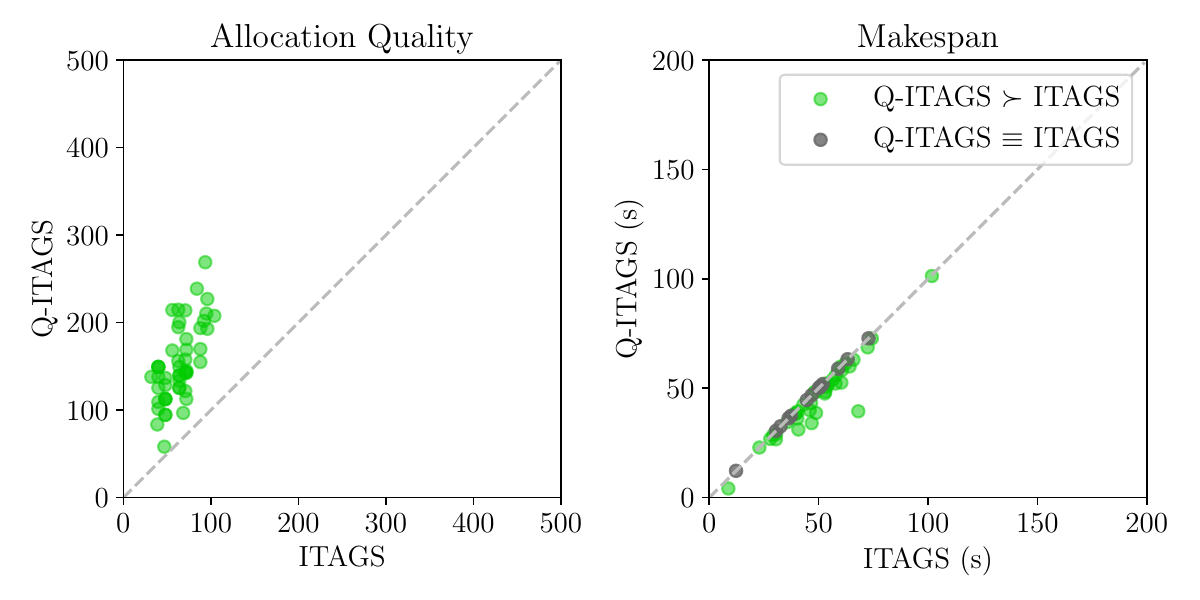}  
    \includegraphics[width=0.34\columnwidth, trim=0 0 0 0, clip]{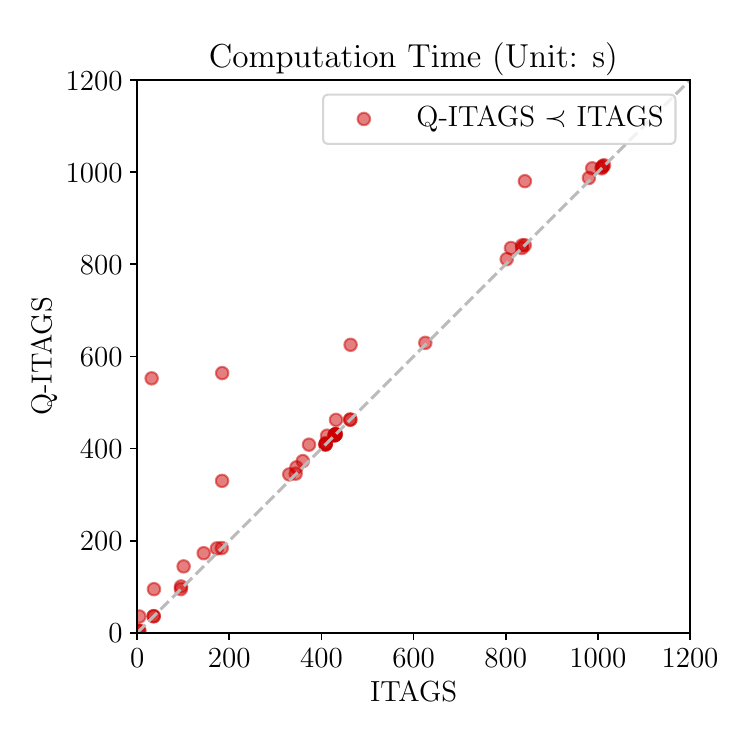}  
\end{center}
\caption{
\small{
Comparison of \algoname' allocation quality, makespan, and computation time against ITAGS. 
\algoname~consistently generates solutions of superior quality (left) while simultaneously ensuring that its makespan is better than or equal to that of ITAGS (middle). Green dots indicate \algoname{} performs better than ITAGS and grey dots indicate \algoname{} and ITAGS perform equally. Larger allocation quality and smaller makespan are desirable. Red dots indicate that \algoname{}' benefits over ITAGS comes at the cost of slightly worse computation time (right).
}
}
\label{fig:comp}
\end{figure}

We evaluated \algoname' performance on 50 problem instances, and contrast it with ITAGS~\cite{Neville2021,neville-2023-d-itags}. We chose to compare our approach with ITAGS as i) it has been shown to perform better than other state-of-the-art time-extended task allocation algorithms, and ii) ITAGS is also a trait-based approach that shares similar assumptions allowing for a more reasonable comparison. A key difference is that ITAGS minimizes makespan subject to allocation requirements, while \algoname{} optimizes allocation performance subject to scheduling requirements. We first computed a solution to each problem using ITAGS and evaluated that solution's allocation quality (using ground-truth trait-quality maps) and makespan. 
We then evaluated \algoname' ability to improve ITAGS' allocation quality when constraining \algoname' time budget to be equal to ITAGS' makespan. 

As shown in Fig. \ref{fig:comp} (\textit{left} and \textit{middle}), \algoname{} consistently improves the allocation quality for all problems, while ensuring that its makespan is either equal to or lower than that of ITAGS. 
This is explained by the fact that \algoname{} can ensure that the makespan doesn't exceed that of ITAGS and yet search for a better solution that maximizes task performance. Indeed, this improved performance comes at the cost of a slight increase in computation time as seen in Fig.~\ref{fig:comp} (\textit{right}). Together, these results demonstrate that \algoname{} outperforms ITAGS in terms of task performance while simultaneously respecting time budget.


\begin{figure}[h]
\begin{center}
    \includegraphics[width=0.55\columnwidth, trim=0 0.6cm 0 0, clip]{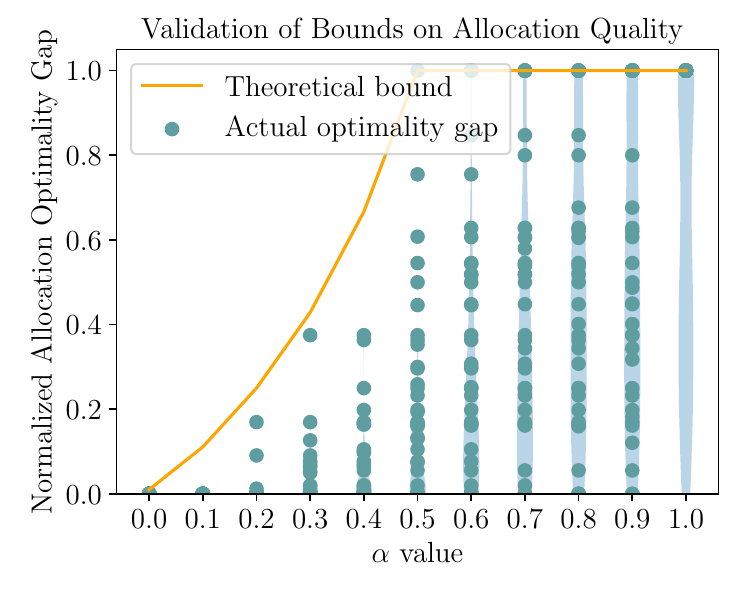}
\end{center}
\caption{
\small{
The theoretical bound consistently holds for varying values of $\alpha$. 
A value of $0$ for a normalized optimality gap represents an optimal allocation, and a value of $1$ represents the worst possible allocation seen within the experiments.
}}
\label{fig:proof} 
\end{figure}

\subsection{Validating Guarantees on Optimality Gaps}
\label{subsec:validation_makespan}

In our second experiment, we empirically examined the validity of our theoretical guarantees on allocation quality from Sec. \ref{sec:theorems}. To this end, we varied $\alpha$ between $[0,1]$ in increments of $0.1$, and solved each of the 50 problem instances from the previous section for each value of $\alpha$. For every combination of problem and $\alpha$ value, we computed the actual normalized optimality gap and the corresponding normalized theoretical bound in Eq.~(\ref{equ:boundsPre}). We found that the optimality gaps consistently respect the theoretical bound across all values of $\alpha$ (see Fig. \ref{fig:proof}). As we expected, the extreme values $\alpha=0$ (ignoring TBO) and $\alpha=1$ (ignoring NAQ) respectively result in the best and worst allocation quality.


\subsection{Learning Trait Requirements for American Football}
\label{subsec:nfl_madden}

\begin{figure}[t]
\begin{center}
    \includegraphics[width=1.0\columnwidth, trim=0 0 0 2cm, clip]{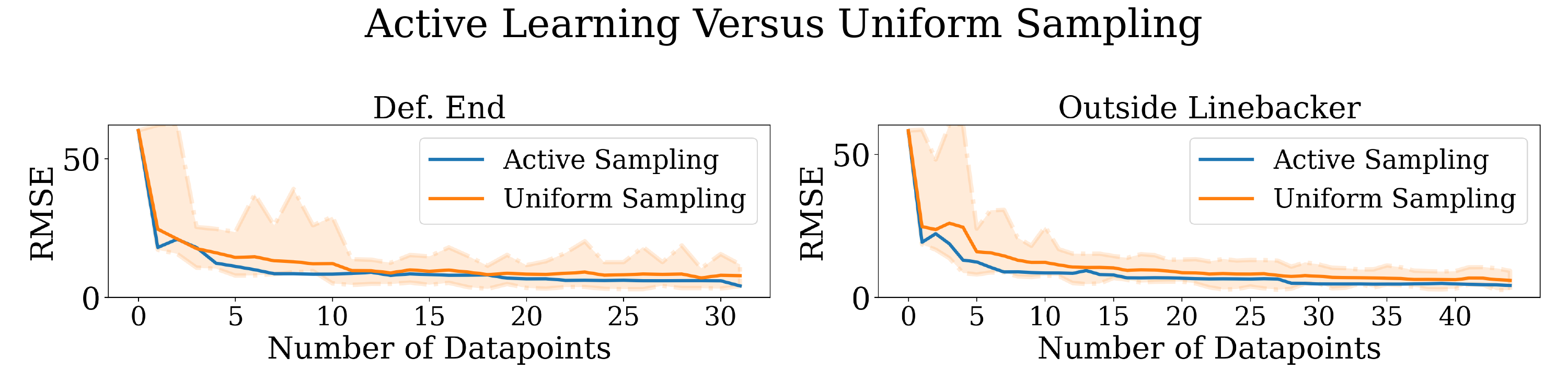}  \\
    \includegraphics[width=1.0\columnwidth, trim=0 0 0 0, clip]{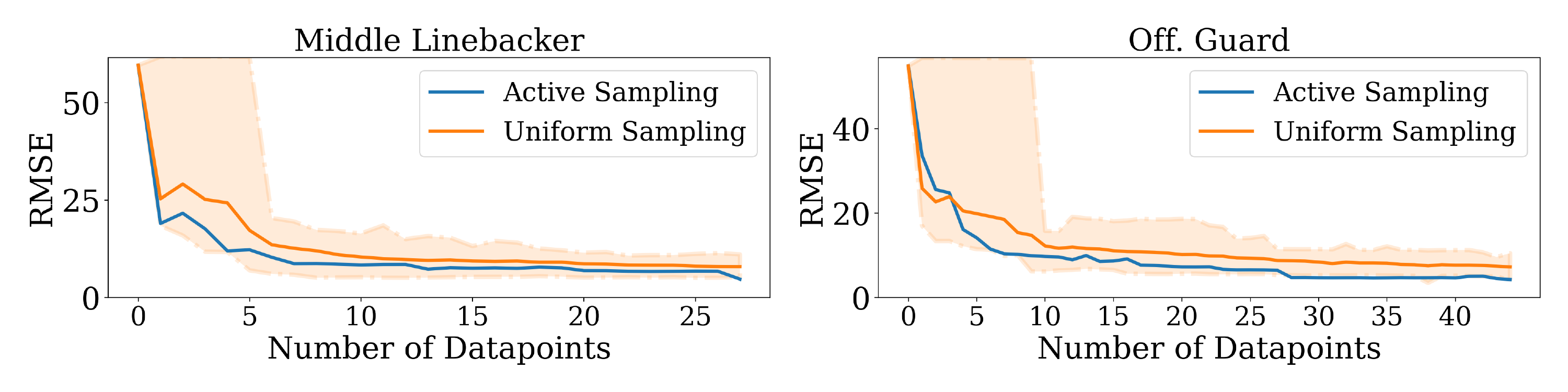}  \\
    \includegraphics[width=1.0\columnwidth, trim=0 0 0 0, clip]{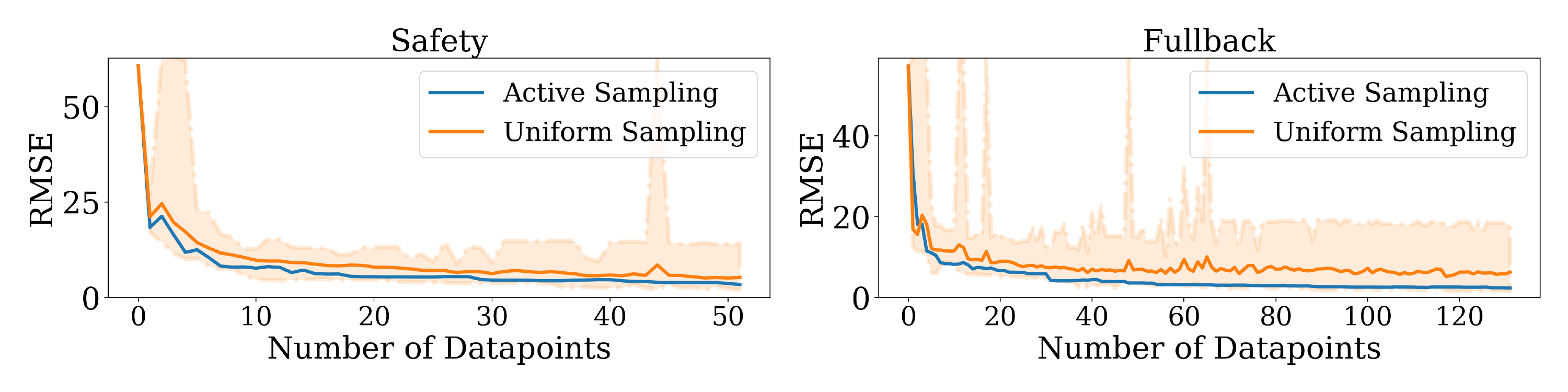} 
\end{center}
\caption{
\small{
\algoname' active learning module consistently outperforms a uniform-sampling baseline in terms of sample efficiency and accuracy when learning trait-quality maps from real-world data. We show prediction errors (RMSE) as a function of number of training samples.
Solid orange curve is the average error of the baseline across 20 random seeds, and the shaded region denotes the maximum and minimum errors. 
}}
\label{fig:madden} 
\end{figure}

Finally, we evaluated our active learning framework's ability to efficiently learn trait-quality maps from real-world data. 
Specifically, we tested the proposed active learning framework using a NFL Madden dataset, in which players are represented with an array of 53 traits, such as toughness and speed, and are also given an effectiveness (i.e., quality) score for a particular position~\cite{Staff_2023}. The ground-truth quality score is given by a weighted sum of the players' traits and the weights used depend on the specific position under consideration. Our active learning module's goal is to efficiently learn the trait-quality maps that predict a given player's quality score for a particular position, given the player's traits. 
 
 We compared \algoname{}' active learning module to a baseline module that uniformly randomly samples training data from the dataset. To quantify and compare their effectiveness, we computed their respective root mean squared prediction error (RMSE).
We computed the performance of the uniform baseline over 20 random seeds to ensure a fair comparison. The results for six different positions are shown in Fig.~\ref{fig:madden}, with shaded regions denoting the maximum and minimum prediction errors of the uniform-sampling baseline.

 As seen in Fig.~\ref{fig:madden}, \algoname{} is consistently more sample efficient than the uniform-sampling baseline, and is capable of learning the trait-quality maps with fewer training samples. In particular, \algoname' performance is highly consistent while the performance of the uniform sampler vary greatly (as indicated by the large variations).  This suggests that \algoname{}' active learning module can select the most informative data points consistently, leading to more efficient and effective learning of trait-quality maps. Given its sample efficiency, \algoname{} can be used in domains in which simulating experiments are computationally expensive or manually labeling data is cumbersome.

\section{Conclusion}
We have formulated a new class of spatio-temporal task allocation problems for heterogeneous robot teams, involving optimization of task performance while respecting a time budget. We also developed a novel algorithm named \algoname{} to solve such problems. Our approach explicitly models coalitions' performance on tasks as continuous trait-quality maps that relate collective capabilities and task performance, and our experiments demonstrate that \algoname{} can learn these maps efficiently from real-world data.
Detailed experiments on an emergency response domain show that \algoname{} can generate solutions of higher quality than a SOTA approach, while still respecting the time budget. We also derived an upper bound on the suboptimality of \algoname{} and demonstrated that it holds in practice. 
Despite all its benefits, \algoname{} has some limitations. For instance, its improved performance over baseline approaches in terms of allocation quality and makespan comes at the cost of a slight increase in computation time (see Fig.~\ref{fig:comp}). Further, \algoname{} doesn't exploit the fact that quality functions are likely to be submodular since adding more capabilities is bound to have diminishing returns beyond a point~\cite{liu2019submodular}. 
Finally, \algoname{} could be extended to handle dynamic scenarios involving robot failures or task changes~\cite{neville-2023-d-itags}. 

%
%

\bibliographystyle{ieeetr}
\bibliography{references.bib}

\end{document}